\documentclass[runningheads]{llncs}

\usepackage[T1]{fontenc}
\usepackage{graphicx}

\usepackage{hyperref}
\usepackage{cite}
\usepackage{amsmath,amssymb,amsfonts}
\usepackage{mathtools}
\usepackage{algorithmic}
\usepackage{physics}
\usepackage{textcomp}
\usepackage{xcolor}
\usepackage{array}
\usepackage{cleveref}

\newtheorem{theo}{Theorem}[section]
\newtheorem{prop}{Proposition}[section]

\newtheorem{defi}{Definition}[section]

\newcommand{\mathcmd}[1]{\ensuremath{#1}}
\providecommand{\Ver}{\mathcmd{\mathsf{Ver}}}
\providecommand{\V}{\mathcmd{\mathsf{V}}}

\providecommand{\Setup}{\mathcmd{\mathsf{Setup}}}
\providecommand{\Com}{\mathcmd{\mathsf{Com}}}

\begin{document}

\title{Commitment Schemes from OWFs with Applications to Quantum Oblivious Transfer}

\author{Thomas Lorünser
\orcidID{0000-0002-1829-4882} \and Sebastian Ramacher\orcidID{0000-0003-1957-3725} \and
Federico Valbusa\orcidID{0009-0006-2065-2673}}

\titlerunning{Commitment Schemes from OWFs with Applications to qOT}
\authorrunning{T. Lorünser, S. Ramacher, and F. Valbusa}

\institute{AIT Austrian Institute of Technology
Vienna, Austria \\
\email{\{thomas.lorünser,sebastian.ramacher,federico.valbusa\}@ait.ac.at}
}

\maketitle

\begin{abstract}
Commitment schemes are essential to many cryptographic protocols and schemes with applications that include privacy-preserving computation on data, privacy-preserving authentication, and, in particular, oblivious transfer protocols. For quantum oblivious transfer (qOT) protocols, unconditionally binding commitment schemes that do not rely on hardness assumptions from structured mathematical problems are required. These additional constraints severely limit the choice of commitment schemes to random oracle-based constructions or Naor's bit commitment scheme. As these protocols commit to individual bits, the use of such commitment schemes comes at a high bandwidth and computational cost.

In this work, we investigate improvements to the efficiency of commitment schemes used in qOT protocols and propose an extension of Naor's commitment scheme requiring the existence of one-way functions (OWF) to reduce communication complexity for 2-bit strings. Additionally, we provide an interactive string commitment scheme with preprocessing to enable a fast and efficient computation of commitments. 

\end{abstract}

\section{Introduction}
In the field of cryptography, the concept of secure communication is fundamental, particularly in scenarios involving multiple parties, where some could behave maliciously. Oblivious transfer (OT)~\cite{STOC:NaoPin99,AC:ABBCP13} results as a fundamental cryptographic primitive that enables two parties, often referred to as the sender and the receiver, to exchange information preserving the privacy of the receiver's choices while ensuring that the sender's data remain hidden except for the selected datum. In the simplest case, for example, the receiver obtains one element from a server with two elements. Protocols can be extended to 1-out-of-$N$ with a set of $N$ elements~\cite{STOC:NaoPin99}. OT is used mainly as a subprimitive for other protocols, such as multiparty computation (MPC)~\cite{DBLP:journals/cacm/Lindell21}, private information retrieval protocols, and simultaneous contract signing protocols~\cite{10.1007/1-4020-8143-X_29}. 

Despite its theoretical significance and practical applications, achieving efficient and provably secure OT protocols often requires striking a delicate balance between computational complexity, communication overhead (e.g., as in~\cite{CCS:KKRT16}), and resilience against various adversarial models, including active attacks and side-channel threats. 

\subsection{Related Work}

In a classical environment, it has been proved that OT implies public-key encryption~\cite{FOCS:GKMRV00}. 
This means that the security of these protocols is based on quantum weak or unproven assumptions. 
Quantum cryptography offers the possibility to build OT protocols whose security is based solely on symmetric key primitives \cite{santos_quantum_2022}.
Specifically, unconditionally binding commitment schemes are essential \cite{lemus2025performancepracticalquantumoblivious} for a practical version of the most relevant quantum oblivious transfer protocol known in the literature \cite{10.1007/3-540-46766-1_29}, even if simulation-based security is not a must \cite{scott_quantum_2002}.

Despite the strong security assumption behind qOT, execution is very slow, also due to the huge amount of data to be exchanged during the post-processing phase required to deal with imperfect quantum channels \cite{santos_private_2022}.
In our experiments, the commitment phase appeared as one substantial bottleneck in qOT execution.
Indeed, a way to commit to many bits (states) and open the commitment of only some of them is needed, ideally with an unconditionally binding commitment scheme whose security is based only on one-way functions.

\subsection{Contribution and Structure}
In this work, we propose a new unconditional binding commitment scheme based only on the existence of OWF, which is communication-wise lighter and still efficient in the case of 2-bit strings. We also present a string commitment of the same type for arbitrarily long strings, whose complexity can be analyzed more easily than the one from Naor. We also propose a generic method to shift the expensive operations of bit commitment schemes in an interactive setting to a preprocessing phase. This method is of particular interest for interactive protocols such as quantum oblivious transfer.

The paper is structured as follows: we begin by outlining the quantum OT protocol that motivates our need for a commitment scheme. Next, we introduce the necessary preliminaries, including the definition and key characteristics of commitment schemes. We then detail Naor's commitment scheme, which forms the foundation for our proposed scheme. Following this, we present a novel, unconditionally binding commitment scheme based on OWF, including an ad-hoc construction for 2-bit strings and an optimization technique to enhance efficiency within the quantum OT context. Finally, we summarize our findings and highlight the achieved improvements.

\subsection{Quantum Oblivious Transfer Protocol}\label{ssec:qot}

We now recall the Quantum Oblivious Transfer Protocol by Bennet et al. \cite{10.1007/3-540-46766-1_29}. The flow of the protocol is shown in \Cref{QOT_prot}.
    \begin{figure}[!h]
        \centering
        \resizebox{\columnwidth}{!}{
        \renewcommand\arraystretch{1.2}
        \begin{tabular}{ l }
            \hspace{5.5 cm} \Large \textbf{$\Pi^{\text{BBCS}}$ protocol}\\   
            \textbf{Parameters:} $n$, security parameter; $\mathcal{F}$ two-universal family of hash functions. \\
            \textbf{S input:} $(m_0,m_1) \in \{ 0,1 \}^l$ (two messages). \\
            \textbf{R input:} $b \in \{0,1\}$ (bit choice). \\
            \textit{BB84 phase:} \\
            \ \ \ 1. S generates random bits $\mathbf{x}^S \xleftarrow{\$} \{0,1\}^n$ and random bases $\mathbf{\theta}^S \xleftarrow{\$} \{+,\times\}^n$. Sends the state \\
            \ \ \ \ \ \  $\ket{\mathbf{x}^S}_{\mathbf{\theta}^S}$ to R. \\
            \ \ \ 2. R randomly chooses bases $\mathbf{\theta}^R \xleftarrow{\$} \{+,\times\}^n$ to measure the received qubits. We denote by $\mathbf{x}^R$  \\
            \ \ \ \ \ \ \ his outputs bits.\\
            \textit{Oblivious key phase}: \\
            \ \ \ 3. S reveals to R the bases $\mathbf{\theta}^S$ used during the \textit{BB84 phase} and sets his oblivious key to \\
            \ \ \ \ \ \ \ $\mathbf{ok}^S  \coloneq \mathbf{x}^S$. \\
            \ \ \ 4. R computes $\mathbf{e}^R = \mathbf{\theta}^R \oplus \mathbf{\theta}^S$ and sets $\mathbf{ok}^R \coloneq \mathbf{x}^R$. \\
            \textit{Transfer phase}: \\
            \ \ \ 5. R defines $I_0 = \{ i: e_i^R=0$\} and $I_1 = \{ i : e_i^R=1$\} and sends $I_b$ to S. \\
            \ \ \ 6. S pick two uniformly random hash functions $f_0, f_1 \in \mathcal{F}$, computes the pair of strings $(s_0, s_1)$ \\
            \ \ \ \ \ \ \ as $s_i = m_i \oplus f_i(ok_{I_{b \oplus i}})$ and sends the pairs $(f_0, f_1)$ and $(s_0, s_1)$ to R. \\
            \ \ \  7. R computes $m_b = s_b \oplus f_b(ok_{I_0}^R)$.\\
            S \textbf{outputs:} $\perp$. \\
            R \textbf{outputs:} $m_b$. \\
        \end{tabular}
        }
        \caption{BBCS QOT protocol.}
        \label{QOT_prot}
    \end{figure}
This protocol requires a natural number $n$ and a two-universal family of hash functions $\mathcal{F}$ as security parameters.
The sender S inputs $(m_0,m_1) \in \{ 0,1 \}^l$, and the receiver R inputs $b \in \{ 0,1 \}$, which represent the two messages and the bit choice of the OT protocol, respectively.

S starts by choosing some basis and sending some qubits R measures with a new random basis. This phase is called the \textit{BB84 phase} since it is the same as the first phase in the quantum key distribution protocol \textit{BB84}.

After this, S reveals his basis and sets his oblivious key $\mathbf{x}^S$ as the set of qubits sent at the very beginning of the protocol. Upon reception of the sender's basis, R can compare the two bases and set his oblivious key as the measured qubits. During this phase, the parties have defined two oblivious keys (one per party), so it is called \textit{Oblivious Key Phase}.

Then, the receiver defines $I_0 = \{ i: {e_i}^R = 0 \}$, the set of indices such that the basis of R and S are the same and $I_1 = I_0^{c}$, the set of indices such that the basis of R and S are different.
R can then decide a bit $b$, which encodes the choice bits of the quantum oblivious transfer protocol, and sends $I_b$ to S. S chooses two random hash functions $f_0, f_1 \in \mathcal{F}$, computes the pair of strings $(s_0, s_1)$ as $s_i = m_i \oplus f_i(\mathbf{ok}^S_{I_{b \oplus i}})$ and sends ($f_0,f_1)$ and $(s_0, s_1)$ to R.
In the end, R is able to compute $m_b = s_b \oplus f_b(\mathbf{ok}^R_{I_0})$. This is possible since $\mathbf{ok}^R_{I_0} = \mathbf{ok}^S_{I_0}$, thanks to the definition of $I_0$. Moreover, this is not the case for $I_1$, so there is no way for S to compute $I_1$. This last phase is called \textit{Transfer Phase}.

\subsubsection{Security}
The security of this protocol is unconditional against a dishonest sender. Indeed, during the protocol, S does not receive any information from R other than $I_b$, which, in the view of S is simply a set of indices where either the basis is the same or different.

However, it is not the same against a dishonest receiver. R can learn both $m_0$ and $m_1$ by delaying his measurements (which should happen in step 2), after step 3. This is called \textit{memory attack} since it requires some sort of memory to hold the received qubits (quantum memory) until step 3.
By delaying the measurements after step 3, R can learn and use the same bases as S to measure the signal and gain knowledge of all the states. This implies that R will know both the input messages of the quantum oblivious transfer protocol.

This attack was discovered directly by the authors of the protocol. They also suggested a way to prevent it: the sender needs to make sure that the receiver measures the received qubits immediately. This can be achieved by requiring R to commit to all measurements (bases and measured states). In particular, each measurement is encoded with a couple of bits, so R has to commit to 2-bit strings. Then, \textit{S} can open a random subset of the values (measurements) to verify R behaved honestly. This brings a good level of security, indeed the possibility that R passes this control is negligible in the size of the set of values S opens.
In \Cref{ModQOT_prot} we can observe the QOT enriched with the bit commitment phase we described above.

\begin{figure}[!h]
        \centering
        \resizebox{\columnwidth}{!}{
        \renewcommand\arraystretch{1.2}
        \begin{tabular}{ l }
            \hspace{5.5 cm} \Large \textbf{$\Pi^{\text{BBCS}}_{\mathcal{F}_{\text{com}}}$ protocol}\\   
            \textbf{Parameters, S input and R input:} Same as in \textbf{$\Pi^{\text{BBCS}}$} protocol. \\
            \textit{BB84 phase:} Same as in \textbf{$\Pi^{\text{BBCS}}$} protocol.\\
            \textit{Cut and choose phase:} \\
            \ \ \ 3. R commits to the bases used and the measured bits, i.e. $\mathbf{com(\theta^R,x^R)}$, and sends to S. \\
            \ \ \ 4. S asks R to open a subset $T$ of commitments (e.g. $n/2$ elements) and receives $\{\theta_i^R,x_i^R \}_{i \in T}$. \\
            \ \ \ 5. In case any opening is not correct or $x_i^R \ne x_i^S$ for $\theta_i^S = \theta_i^S$, abort. Otherwise, proceed. \\
            \textit{Oblivious key phase}: Same as in \textbf{$\Pi^{\text{BBCS}}$} protocol.\\
            \textit{Transfer phase}: Same as in \textbf{$\Pi^{\text{BBCS}}$} protocol.\\
            S \textbf{output} \textbf{and} R \textbf{output:} Same as in \textbf{$\Pi^{\text{BBCS}}$} protocol.
        \end{tabular}
        }
        \caption{BBCS QOT protocol with the addition of a commitment functionality $\mathcal{F}_{com}$.}
        \label{ModQOT_prot}
    \end{figure}
There exist proofs of security for a practical version of the quantum oblivious transfer protocol enriched with the commitment scheme \cite{lemus2025performancepracticalquantumoblivious}.
All these proofs make use of unconditional binding commitment schemes. Moreover, the commitment scheme used in these proofs should be based on OWF only. Indeed, we want to use the qOT protocol to avoid weakness against quantum cryptanalysis and having unproven and not well-established hypotheses. If we had used a post-quantum commitment scheme, we could have directly employed a post-quantum OT protocol.

\section{Preliminaries}
\subsection{Definitions}
In simple terms, a commitment scheme is a protocol that allows a party called \textit{prover} to commit to a hidden value and send it to another party called \textit{verifier}. Once the value has been sent the verifier should not be able to check it until the sender gives their permission. Also, the prover should not be able to change its value once it has been sent. We can think of the protocol as a box: the sender sends their committed message to the verifier hidden in a locked box. Once the verifier has the locked box, the sender is not able to change the value anymore. Moreover, the verifier is not able to unlock and open the box, since the sender has the key. Finally, the sender sends the key to the verifier so the latter can check the correctness of the claimed value.

Following \cite{AC:JKPT12}, we give the formal definition of a commitment scheme:
\begin{defi}[Commitment Scheme]
    A \textit{commitment scheme} consists of three polynomial-time algorithms $(\Setup,\Com,\Ver)$ such that:
    \begin{itemize}
        \item[-] $\Setup(1^n)$ takes as input $1^n$, for a security parameter $n$, and outputs some public parameter $pk$ as a public commitment key, i.e., $pk \leftarrow \Setup(1^n)$;
        \item[-] $\Com(pk,m)$ takes as input a public key $pk$, and a message $m$. It outputs a commitment $c$, and a reveal value $d$, i.e., $(c, d) \leftarrow \Com(pk, m)$;
        \item[-] $\Ver(pk,m,c,d)$ takes as input a public key $pk$, a message $m$, a commitment $c$, and a reveal value $d$. It returns 1 or 0 to accept or reject, respectively, i.e., $b \leftarrow \Ver(pk, m, c, d)$, where $b \in \{0, 1\}$. 
    \end{itemize}
\end{defi}

Moreover, it must satisfy the following basic requirements:
\begin{itemize}
    \item[-] \textit{Perfect completeness}: the verification algorithm $\Ver$ outputs 1 whenever the inputs are computed honestly, i.e. for all messages $m$,
    \[ \mathbb{P} [  \Ver(pk,m,c,d) = 1 \mid pk \leftarrow \Setup(1^n), (c,d) \leftarrow \Com(pk,m)  ] = 1.\]
\end{itemize}
We say that a commitment scheme is \textit{interactive} if the commitment algorithm requires some exchange of messages between two or more parties. Otherwise, the commitment scheme is said to be \textit{non-interactive}.
\subsubsection{Security}
The two fundamental security properties of commitment schemes are the \textit{hiding} and the \textit{binding}, as introduced above. Let us give a formal definition of these: 
\begin{defi}[Binding Property]
    A commitment scheme $(\Setup,\Com,\Ver)$ is \textit{statistically} (respectively, \textit{computationally}) \textit{binding} if no (respectively, PPT) forger $\mathcal{P}^*$ can come up with a commitment and two different openings with noticeable (non-obvious) probability, i.e., for every (respectively, PPT) forger $\mathcal{P}^*$, and every two distinct messages $m$, $m'$ from the message space, there exists a negligible function $\epsilon (n)$ such that
    \[ \mathbb{P}\left [ \V_1 = \V_2 = 1|pk \hspace{-1mm} \leftarrow \Setup(1^n), (c,m,d,m',d') \leftarrow \mathcal{P^*}(pk) \right ] \le \epsilon(n),\]
    where $\V_1 = \Ver(pk,m,c,d)$, $\V_2 = \Ver(pk,m',c,d'), (c,d)$ $\leftarrow\Com(pk,m),(c,d')\leftarrow\Com(pk,m'), pk \leftarrow \Setup(1^n)$.
    
    The scheme is \textit{perfectly binding} if, with overwhelming probability over the choice of the public key $pk \leftarrow \Setup(1^n)$, it holds
    \[ \left (\Ver(pk,m,c,d) = 1 \land \Ver(pk,m',c,d')=1 \right ) \Rightarrow m = m'.\]
\end{defi}

\begin{defi}[Hiding Property]
    A commitment scheme $(\Setup,\Com,\Ver)$ is \textit{statistically} (respectively, \textit{computationally}) \textit{hiding} if no (respectively, PPT) distinguisher $\mathcal{V}^*$ is able to distinguish $(pk, c)$ and $(pk, c')$, for two distinct messages $m$, $m'$, with non-negligible advantage. This means that for every (respectively, PPT) distinguisher $\mathcal{V}^*$, and every two distinct messages $m$, $m'$ from the message space, there exists a negligible function $\epsilon(n)$ such that the statistical distance between $(pk, c)$ and $(pk, c')$ is $\epsilon (n)$, i.e.
    \[ \Delta \left ( (pk,c),(pk,c') \right ) \le \epsilon(n),\]
    where $(c,d) \leftarrow \Com(pk,m)$, $(c',d') \leftarrow \Com(pk,m')$, and $pk \leftarrow \Setup(1^n)$.\newline
    The scheme is \textit{perfectly hiding} if $(pk,c)$ and $(pk,c')$ are equally distributed.
\end{defi}

If one of these two properties holds statistically or perfectly, we say that it holds \textit{unconditionally}.
There exists an important negative result regarding these two properties. Namely, a commitment scheme cannot be both unconditionally hiding and unconditionally binding.
Thus, at least one of both the prover and verifier must be at most computationally bounded.

\subsection{Commitment Scheme for qOT}
As discussed in Section \ref{ssec:qot}, an unconditionally binding commitment scheme is required, ideally based on the existence of OWF.
There are two possible schemes satisfying this: one is based on domain-increasing hash functions \cite{halevi_practical_1996} and the other is the commitment scheme from Naor \cite{JC:Naor91}.
We will consider using the latter, which can be less time-consuming and gives the same level of security as the former approach with fewer assumptions.

Indeed, to prove a commitment scheme based on domain-increasing hash functions is statistically binding, the random oracle model assumption must be employed.
This is not necessary for Naor's commitment scheme.
Moreover, Naor's approach can also be faster if AES in counter mode is exploited as a secure pseudorandom number generator (PRNG), given the size of the commitment \cite{faz_hernandes_performance_2018}.
The only advantage of the hash approach is that it gives a non-interactive commitment scheme. However, the commitment must be executed, as already seen previously, in a protocol requiring interaction, so non-interactivity is of no benefit in this case.

\subsection{Naor's Bit Commitment Scheme}
The bit commitment scheme from Naor is described, as the proofs of security for the new schemes will be similar. \newline
Let $n$ be a security parameter. In particular, it guarantees the security of the used PRG for seeds of length $n$. Then the scheme works as follows:
\begin{itemize}
    \item[-] $\Setup(1^n)$: the setup algorithm chooses any pseudorandom generator $G$ that stretches a random seed with length $n$ to $3n$ bits, i.e., $G: \{0,1\}^n \rightarrow \{0,1\}^{3n}$;
    \item[-]$\Com(G,b)$: this algorithm has two phases:
    \begin{itemize}
        \item Commit phase: 
        \begin{itemize}
            \item[1.] $\mathcal{V}$ selects a random bit string \textbf{r} of length $3n$ and sends it to $\mathcal{P}$;
            \item[2.] After receiving \textbf{r}, $\mathcal{P}$ selects a random seed $\mathbf{x} \in \{0,1\}^n $ and returns the commitment string \textbf{c} to $\mathcal{V}$, where 
            \[ \mathbf{c} = \begin{cases}
                G(\mathbf{x}) \quad \quad \ \ \text{ if } b=0 \\ 
                G(\mathbf{x}) \oplus \mathbf{r} \quad \text{ if } b = 1
            \end{cases}\]
        \end{itemize}
        \item Reveal phase: to open the commitment, $\mathcal{P}$ sends $b$ and \textbf{x} to $\mathcal{V}$.
    \end{itemize}
    \item[-] $\Ver(G,\mathbf{r,c},b,\mathbf{x})$: $\mathcal{V}$ verifies that the values $b, \mathbf{x}$ and \textbf{r} match the previously given commitment.
\end{itemize}

Now, the properties of this commitment scheme are formally proven:
\begin{itemize}
    \item \textit{Computational Hiding:} we have to prove that no probabilistic polynomial-time (PPT) distinguisher $\mathcal{V^*}$ is capable of distinguishing between $(G, \mathbf{c})$ and $(G, \mathbf{c}')$ for two distinct bits $b$ and $b'$ with a non-negligible advantage. Specifically, we will demonstrate that if $G$ is a secure pseudorandom generator (PRG), then the scheme possesses the hiding property.
    To argue by contradiction, suppose the scheme is not hiding. This implies the existence of an adversary $\mathcal{A}$ who can distinguish between $(G, \mathbf{c})$ and $(G, \mathbf{c}')$ for two different bits $b$ and $b'$ with a non-negligible probability.
    Given this situation, we can construct another adversary $\mathcal{B}$ that can distinguish between a truly random string $\mathbf{r} \xleftarrow{\$} \{0,1\}^{3n}$ and a pseudorandom string $ G(\mathbf{s})$ for a seed $\mathbf{s} \xleftarrow{\$} \{0,1\}^n$.

    To model this attack, we refer to  \Cref{NaorCShiding}. Notably, we recognize that the adversary $\mathcal{A}$ can break the hiding property. In other words, $\mathcal{A}$ is able to distinguish between $G(\mathbf{s})$ and $G(\mathbf{s}) \oplus \mathbf{r}'$ with a non-negligible probability in $n$, where $\mathbf{s}$ is a random string chosen from $\{0,1\}^n$ and $\mathbf{r}' \in \{0,1\}^{3n}$ is chosen by $\mathcal{A}$.
    
    Now, let us assume the challenger selects $\mathbf{r}$ to be a truly random string, i.e., $\mathbf{r} \xleftarrow{\$} \{0,1\}^{3n}$.
    Under this assumption, it becomes crucial to analyze the behavior of $\mathcal{A}$ and understand how it exploits the properties of $G$ to distinguish between the given instances.
    \begin{figure}[!h]
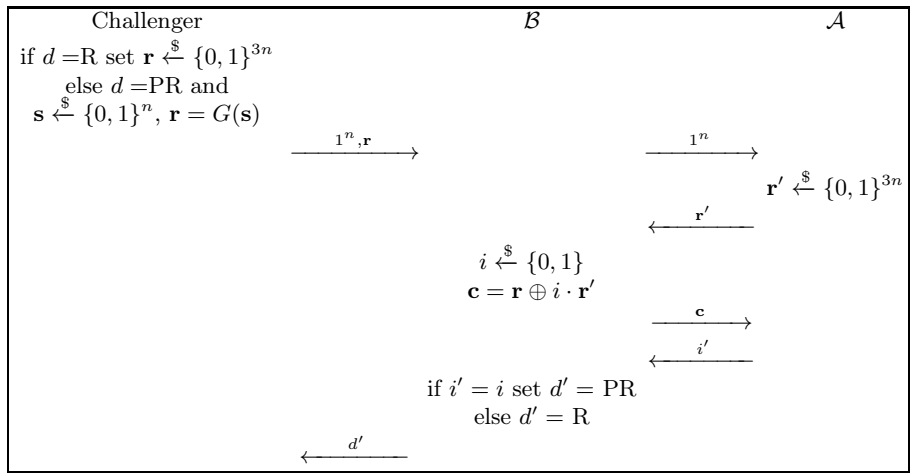

        \centering
        \resizebox{\columnwidth}{!}{
        \begin{tabular}{|*{1}{>{\centering\arraybackslash}p{.3\linewidth}} c c c c |}
            \hline
            Challenger & & $\mathcal{B}$ & & $\mathcal{A}$\\
             if $d=$R set $\mathbf{r} \xleftarrow{\$} \{0,1\}^{3n}$ & & & &\\
             else $d=$PR and $\mathbf{s} \xleftarrow{\$} \{0,1\}^{n}$, $\mathbf{r} = G(\mathbf{s})$ & & & & \\
             & $\xrightarrow{\hspace{15pt} 1^n,\mathbf{r} \hspace{15pt}}$ & &  $\xrightarrow{\hspace{15pt} 1^n \hspace{15pt}}$ & \\
             & & & & $\mathbf{r'} \xleftarrow{\$} \{0,1\}^{3n}$ \\
             & & & $\xleftarrow{\hspace{15pt} \mathbf{r'} \hspace{15pt}}$ & \\
             & & $i \xleftarrow{\$} \{0,1\}$ & & \\
             & & $\mathbf{c} = \mathbf{r} \oplus i \cdot \mathbf{r'}$ & & \\
             & & & $\xrightarrow{\hspace{15pt} \mathbf{c} \hspace{15pt}}$ & \\
             & & & $\xleftarrow{ \hspace{15pt}  i' \hspace{15pt}} $ & \\
             & & if $i'=i$ set $d'=$ PR & & \\
             & & else $d'=$ R & & \\
             & $\xleftarrow{\hspace{15pt} d' \hspace{15pt}}$ & & & \\
        \hline
        \end{tabular}
        }
        \caption{Attack on the PRG $G$ if the hiding property of Naor's bit commitment scheme can be broken.}
        \label{NaorCShiding}
    \end{figure} 
Then, $\mathbf{r} \oplus i \cdot \mathbf{r'}$ remains random for $i \in \{0,1\}$, so $\mathcal{A}$ cannot distinguish it with non-negligible probability.
        
Specifically, in this case, we have:
\[\mathbb{P}\left [i = i' \big | \mathbf{r} \xleftarrow{\$} \{0,1\}^{3n} \right ] = \frac{1}{2}.\]
Now, consider the scenario where the challenger instead chooses $\mathbf{r}$ as a pseudorandom string (i.e., $\mathbf{r} = G(\mathbf{s})$, where $\mathbf{s} \xleftarrow{\$} \{0,1\}^{n}$). In this situation, $\mathbf{c} = G(\mathbf{s}) \oplus i \cdot \mathbf{r'}$. Under this assumption, when $i=0$, $\mathbf{c}=G(\mathbf{s})$, which is a pseudorandom string generated by $G$ using a random seed $\mathbf{s}$.
Instead, when $i=1$, $\mathbf{c}$ is the XOR of a pseudorandom string $G(\mathbf{s})$ and a random string ($G(\mathbf{s}) \oplus \mathbf{r'}$).
This is precisely the case where $\mathcal{A}$ can distinguish with a non-negligible advantage the two cases. \newline
Thus, with a non-negligible advantage, $\mathcal{A}$ will correctly identify $i$, that is:
\[\mathbb{P}\left [i = i' \big | \mathbf{r} = G(\mathbf{s}), \ \mathbf{s} \xleftarrow{\$}\{0,1\}^n\right ] = \frac{1}{2} + non\text{-}neg(n).\]
Recall that the ultimate goal is to prove that $\mathcal{B}$ can distinguish a pseudorandom string from a random string with non-negligible probability, i.e.,
\[|\mathbb{P}[\mathcal{B} \text{ outputs PR}| \mathbf{r} \text{ is random}] - \mathbb{P}[\mathcal{B} \text{ outputs PR}| \mathbf{r} \text{ is pseudorandom}]| \]\[= non\text{-}neg(n).\]
In this case, the previous difference is:
\[ \left | \mathbb{P}\left [i = i' \big | \mathbf{r} = G(\mathbf{s}), \ s \xleftarrow{\$}\{0,1\}^n\right ] - \mathbb{P}\left [i = i' \big | \mathbf{r} \xleftarrow{\$} \{0,1\}^{3n} \right ] \right | =\] \[ \left | \left (\frac{1}{2} + non\text{-}neg(n) \right ) - \frac{1}{2} \right | = non\text{-}neg(n).\]

This result confirms that $\mathcal{B}$ can indeed distinguish between a pseudorandom string and a random string with a non-negligible probability, confirming the computational hiding property of the scheme.
\item \textit{Statistical Binding:} we must establish that for every potential forger $\mathcal{P}^*$ and for any two distinct messages $m$ and $m'$ from the message space, there exists a negligible function $\epsilon(n)$ such that
\[ \mathbb{P}[ \V_1 = \V_2 = 1|pk \leftarrow \Setup(1^n),(c,m,d,m',d') \leftarrow \mathcal{P^*}(pk)] \le \epsilon(n),\]
where $\V_1 = \Ver(pk,m,c,d)$, $\V_2 = \Ver(pk,m',c,d')$, $(c,d) \leftarrow \Com(pk,m)$, $(c,d') \leftarrow \Com(pk,m')$, and $ pk \leftarrow \Setup(1^n)$.\newline Considering the commitment notation, this requirement translates to
\[ \mathbb{P}\left [ G(\mathbf{x}) = \mathbf{c} = G(\mathbf{x'}) \oplus \mathbf{r} \mid \mathbf{r} \in \{0,1\}^{3n}, \ \mathbf{x},\mathbf{x}' \in \{0,1\}^n  \right ]. \]
It is important to note that the set $\{G(\mathbf{x}) \mid \mathbf{x} \in \{0,1\}^n\}$ contains at most $2^n$ elements. Consequently, the set $\{G(\mathbf{x}) \oplus G(\mathbf{x'}) \mid \mathbf{x} \in \{0,1\}^n, \mathbf{x'} \in \{0,1\}^n\}$ can have at most $2^{2n}$ elements. Given this, the probability that a randomly chosen string $\mathbf{r} \in \{0,1\}^{3n}$ falls within this set is at most \[\frac{2^{2n}}{2^{3n}} = 2^{-n}.\]

Therefore, this is the maximum probability with which a forger can identify two values that generate two different valid openings for the commitment. This extremely low probability, decreasing exponentially with $n$, ensures that the statistical binding property is verified, effectively preventing the forger from finding such values and thus maintaining the integrity and security of the commitment scheme.

\subsection{On Naor's String Commitment Scheme}
In \cite{JC:Naor91} Naor presented also a string version of his commitment scheme. This scheme is not analyzed here because not useful. What is interesting about this scheme, related to this work, is that the complexity to commit to a string of $t$ bits, communication-wise, is $4q+n+t$, where $n$ is the number of security bits regarding the binding property, $q$ must satisfy $q> \frac{3n}{\log \left (\frac{2}{2-\epsilon}\right )}$ and $\epsilon q$ is the minimum distance in a set of $2^t$ binary vectors of length $q$. To keep $q$ small, the set of $2^t$ vectors of length $q$ with the highest minimum distance should be picked. In general, this problem is known to be hard. However, it can be easily proved that the minimum distance in a set of at least 3 binary vectors of length $n$ is at most $\frac{2}{3}n$, leading to a total complexity of at least $7.8n + 2$ per 2-bit string we have to commit to.

As pointed out at the end of \cite{JC:Naor91}, Joe Kilian has suggested a slightly different method for amortizing the communication complexity when we have to commit to a bit string. In particular, we can commit to a seed \textbf{s} by committing to each of its bits separately, with the Naor's bit commitment scheme, and then commit to the string $(b_1, b_2, \dots, b_t)$ by providing its XOR with the pseudorandom sequence generated by \textbf{s}.
In this case we need two pseudorandom number generators: $G_1:\{0,1\}^n \rightarrow \{0,1\}^{3n}$ and $G_2: \{0,1\}^{|\mathbf{s}|} \rightarrow \{0,1\}^t$.\newline
However, this method starts to be effective when $t$ is at least $n^2$, since the communication cost is always at least $3n^2$, independently of the value of $t$. The hiding and binding properties of this protocol come directly from Naor's bit commitment scheme.

\end{itemize}

\section{New String Commitment based on Naor}

Given Naor's bit commitment scheme, the commitment $\mathbf{c}$ for a bit $b$ is calculated as $\mathbf{c} = G(\mathbf{x}) \oplus b \cdot \mathbf{r}$. From this, one can think that if $\mathcal{P}$ (the prover) needs to commit to more than one bit at a time, say a string of bits $(b_1, \dots, b_t)$, it can compute the commitment as $\mathbf{c} = G(\mathbf{x}) \oplus \sum_{i=1}^t b_i \cdot \mathbf{r}_i$. Indeed, the bit commitment is a special case of this with $t=1$.

However, in this case, the vectors $\{\mathbf{r}_i\}_{i=1}^t$ must be linearly independent. Indeed, if they are not linearly independent, then there exists a non-zero linear combination that gives the zero vector, namely $\sum_{i=1}^t b_i\mathbf{r}_i = \mathbf{0}$ with some string $(b_1, b_2, \dots, b_n)$ where $b_i \ne 0$ for at least one index $i$. But then, committing to the bit string $(b_1, b_2, \dots, b_n)$ is the same as committing to the string $(0,0, \dots, 0)$, so the statistical binding property is certainly not satisfied since the committer can always commit to $(b_1, b_2, \dots, b_n)$ and open to $(0,0, \dots, 0)$ and vice versa.

Also, if $\mathcal{V}$ has to choose and send all these vectors to $\mathcal{P}$, the communication cost can become very high. If there exists a ``good" way to generate these vectors starting from one initial vector $\mathbf{r}_1$ chosen by $\mathcal{V}$, this communication cost can be significantly reduced. This makes committing to a string of $t$ bits much more efficient compared to committing to $t$ bits separately, especially in terms of communication costs. For the moment, assume that the prover and the receiver have a random oracle $\mathcal{O}$ that returns random linearly independent vectors given as input a specific one.

To ensure a security level of $n$ bits for the binding property, we need to increase the length of the pseudorandom number generator output by $t$ bits. The protocol's flow is detailed in \Cref{ourBSC}.
\begin{figure}[!h]
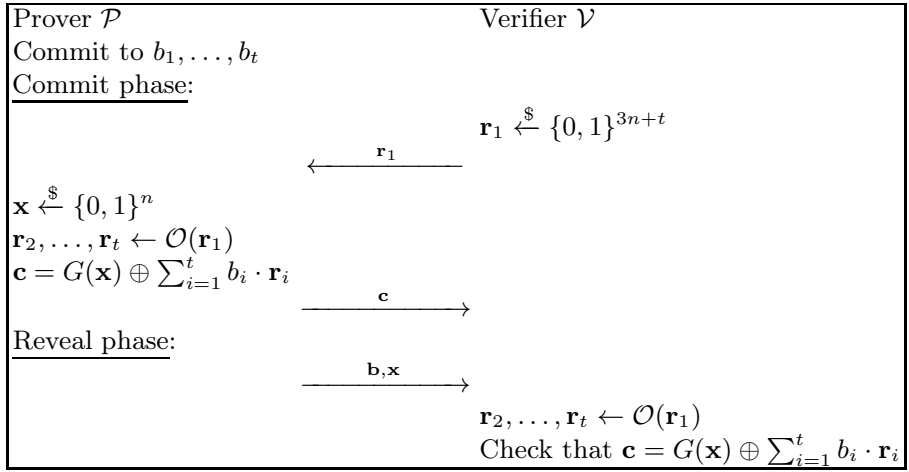

        \centering
        \resizebox{\columnwidth}{!}{
        \begin{tabular}{|l c l|}
            \hline
            Prover $\mathcal{P}$ & & Verifier $\mathcal{V}$ \\
            Commit to $b_1, \dots, b_t$ & & \\
            \underline{Commit phase}: & & \\
            & & $\mathbf{r}_1 \xleftarrow{\$} \{ 0,1 \} ^{3n+t}$ \\
            & $\xleftarrow{\hspace{20 pt}  \mathbf{r}_1 \hspace{20 pt}}$ & \\
            $\mathbf{x} \xleftarrow{\$} \{ 0,1 \} ^{n}$ & &\\
            $\mathbf{r}_2, \dots, \mathbf{r}_t \leftarrow \mathcal{O}(\mathbf{r}_1)$ & &\\
            $\mathbf{c} = G(\mathbf{x}) \oplus \sum_{i=1}^t b_i \cdot \mathbf{r}_i$ & &\\
            & $\xrightarrow{\hspace{24 pt} \mathbf{c} \hspace{24 pt}}$ & \\
            \underline{Reveal phase}: & & \\
            & $\xrightarrow{\hspace{20 pt} \mathbf{b}, \mathbf{x} \hspace{20 pt}}$ & \\
            & & $\mathbf{r}_2, \dots, \mathbf{r}_t \leftarrow \mathcal{O}(\mathbf{r}_1)$ \\
            & & Check that $\mathbf{c} = G(\mathbf{x}) \oplus \sum_{i=1}^t b_i \cdot \mathbf{r}_i$ \\
        \hline
        \end{tabular}
        }
        \caption{Natural extension of Naor's bit CS.}
        \label{ourBSC}
    \end{figure}
It is possible to prove that this commitment scheme is computationally hiding and statistically binding:    
\begin{itemize}
    \item \textit{Computational Hiding}:
    the poof of the computational hiding property can be obtained naturally by modifying that of Naor's bit commitment scheme. In particular, the attack is the same as the one described in \Cref{NaorCShiding}, just changing the fact that instead of committing to a bit, a bit string is committed to and the adversary $\mathcal{A}$ can guess the right values for $(b_1, \dots, b_t)$ with a probability of $\frac{1}{2^t} +$ $non$-$neg(n)$ when is given $\mathbf{r}=G(\mathbf{s})$, and with probability $\frac{1}{2^t}$ when $\mathbf{r} \xleftarrow{\$} \{0,1\}^{3n}$;  
   
    \item \textit{Statistical Binding}: suppose $\mathcal{P}$ wants to open the commitment to a different value $\mathbf{b}' \ne \mathbf{b}$ after already committing to $\mathbf{b}$. To do this, $\mathcal{P}$ needs to find a new $\mathbf{x}'$ such that $\mathbf{c} = G(\mathbf{x}') \oplus \sum_{i=1}^t b_i' \mathbf{r}_i$. We know that for any fixed value of $\mathbf{b}'$, the chance of successfully finding such an $\mathbf{x}'$ is very low, specifically bounded by $2^{-n-t}$. This can be shown by following the proof of the binding property in Naor's bit commitment scheme and adjusting the length of the pseudorandom string generated by $G$.
   
    Thus, the probability that $\mathcal{P}$ can find one such value to cheat is bounded by $2^{t} \cdot 2^{-n-t} = 2^{-n}$ since all possible values that allow $\mathcal{P}$ to cheat are $2^{t}-1$ (that is, all the possible values $\mathbf{b'} \ne \mathbf{b}$) and the probability for all is bounded by $2^{-n-t}$.
\end{itemize}

\subsection{On the Algorithm to Generate Some Linearly Independent Vectors From a Random One}
The previous protocol exploits a random oracle to generate some linearly independent vectors starting from a single random vector. In practice, however, it is not possible to rely on this, since the existence of random oracles is an idealized assumption. Thus, a deterministic algorithm to generate such vectors is needed. However, the generated vectors must be, in a certain sense, uniformly distributed (which is the case when they are randomly picked) between the linearly independent ones.

For example, suppose we start with one randomly chosen non-zero vector in $\{0,1\}^n$ with $n \geq 2$, and we pick the second vector fixing the first coordinate of the first vector, whose value is 1 and flipping all the other bits.
In this manner, we obtain two linearly independent vectors, as the second vector is not the zero vector (it retains a 1 from the first vector) and it is different from the first vector (at least one bit is flipped, since $n \geq 2$). However, with this approach, we know that the bitwise XOR between the two vectors is a vector with all coordinates set to 1 except for one coordinate, which is 0.

If we suppose the usage of an algorithm like this in a scheme like the one described in \Cref{ourBSC}, the statistical binding property is compromised. In particular, the binding property does not hold unconditionally anymore, it depends on the pseudorandom number generator. Indeed, if $G(\mathbf{x})$ is defined such that all vectors in the codomain with all 1s except for one 0 are images of some input $\mathbf{x}$, the prover can always find, by brute-force, an $\mathbf{x}'$ such that $G(\mathbf{x}) \oplus \sum_{i=1}^2 \mathbf{r}_i = G(\mathbf{x}')$, allowing the prover to commit to (1,1,0,$\dots$,0) and to open the commitment to the 0 string.

Therefore, all generated vectors and their linear combinations should be ``uniformly distributed''. Uniform distribution, in this case, means that fixed a non-zero linear combination of the vectors generated, by changing the value to the initial vector we have to obtain every time a different result.

Mathematically speaking, this means that the linear system given by $\mathbf{s} = \sum_{i=1}^t b_i \cdot \mathbf{r}_i$, must have a unique solution in the indeterminates describing the generated vectors, for every coefficient combination of $b_1, \dots, b_t$ and for every value of $\mathbf{s}$.

We were not able to find an algorithm satisfying such a strong property. However, vectors with a slightly weaker property can still be generated, which suffices to prove the binding property of the scheme.

To show and discuss the new protocol, we need a couple more definitions and results:
 \begin{defi}[Circulant Matrix]
     A circulant matrix is a square matrix in which each row is obtained by rotating one element to the right of the preceding row.
 \end{defi}

 \begin{defi}[Associated Polynomial of a Circulant Matrix]
     Let $C$ be the circulant matrix obtained rotating the vector $(a_0, a_1, \dots, a_{n-1})$. Then the polynomial $f = \sum_{i=0}^{n-1} a_ix^i$ is called the associated polynomial of $C$.
 \end{defi}
\begin{theo}\label{oddWeight}
    Let $C$ be a circulant matrix with entries in $\mathbb{F}_2$ of order $n=2^t$ for some $t \ge 0$. Then the following conditions are equivalent:
    \begin{itemize}
        \item The rank of $C$ is $n$;
        \item The degree of the greatest common divisor of $x^n - 1$ and the associated polynomial of $C$ is 0.
    \end{itemize}
\end{theo}
\begin{proof}
Let $g$ be the greatest common divisor of $f$ and $x^n - 1$. First, assume $f \neq 0$, since if  $f = 0$, the corresponding circulant matrix is the zero matrix. Now, suppose $g \neq 0$. By definition, $g$ divides $x^n - 1 = x^{2^t} - 1 = (x - 1)^{2^t} = (x \oplus 1)^{2^t}$, due to the fact that we are operating in $\mathbb{F}_2$, where subtraction is equivalent to addition (i.e., $- = \oplus$) and in finite fields with characteristic $p$, powers of $p$ can be distributed over sums. Consequently, $g$ divides $(x \oplus 1)^{2^t}$. Since $x \oplus 1$ is irreducible, we have either $g = 1$ or  $x \oplus 1$ divides $g$.

We now proceed to prove both implications:

\begin{enumerate}
    \item $\deg(g) = 0$ $\Rightarrow$ $\text{rank}(C) = n$:
    
    Suppose, for the sake of contradiction, that $\text{rank}(C) < n$. This implies the existence of a non-zero linear combination of the matrix vectors that results in the zero vector. Each vector can be identified with a polynomial in $ \mathbb{F}_2[x] / (x^n - 1) $. Specifically, the polynomial corresponding to the first vector is the polynomial $ f $; the second vector corresponds to $ x \cdot f \mod{x^n - 1} $, and so on. Finding a non-zero linear combination that generates $ \mathbf{0} $ translates into finding a subset $ I \subseteq [n-1] \cup \{0\} $ such that $ \sum_{i \in I} \mathbf{x}_i = \mathbf{0} $. In polynomial terms, this means $ \sum_{i \in I} x^i \cdot f \equiv 0 \mod{x^n -1} $, which implies that $ x^n - 1 $ divides $ f \cdot \sum_{i \in I} x^i $. Since $ x^n - 1 = (x \oplus 1)^{2^t} $, two possibilities arise: either $ x \oplus 1 $ divides $ f $, which implies that $ x \oplus 1 $ divides both $ f $ and $ x^n - 1 $, leading to $ \deg(g) > 0 $, contradicting our initial assumption; or $ (x \oplus 1)^{2^t} $ divides $ \sum_{i \in I} x^i $. Given that $ I \subseteq [n-1] \cup \{0\} $, this is only possible if $ I = \emptyset $, which contradicts our assumption of a non-zero linear combination.

    \item $\text{rank}(C) = n$ $\Rightarrow$ $\deg(g) = 0$:
    
    Suppose, for the sake of contradiction, that $ \deg(g) \neq 0 $. According to the earlier discussion, this implies $x \oplus 1 $ divides $g$. Since $ g $ divides $ f $ by definition, $ f $ must have a root in 1. This means that $ f $ has an even number of terms. If this is the case, it is easy to check that the sum of all vectors in the matrix is $ \mathbf{0} $, implying that the rank of the matrix is strictly less than $ n $. Therefore, $ g = 1 $, and thus its degree is 0.
\end{enumerate}
\end{proof}
So, for the generation of vectors, we can consider the following algorithm:  
\begin{align} \label{finalAlg}
    \mathbf{x} &= \begin{pmatrix}
        x_1 \\
        x_2 \\
        \vdots \\
        x_{n-1 } \\
        x_n
    \end{pmatrix}
    \rightarrow \begin{matrix}
        x_1 & x_n & \cdots_{\phantom{n}} & x_{3\phantom{-n}} & x_2 \\
        x_2 & x_1 & \cdots_{\phantom{n}} & x_{4\phantom{-n}} & x_3 \\
        \vdots & \vdots & \ddots_{\phantom{n}} & \vdots_{\phantom{-n}} & \vdots \\
        x_{n-1} & x_{n-2} & \cdots_{\phantom{n}} & x_{1\phantom{-n}} & x_n \\
        x_n & x_{n-1} & \cdots_{\phantom{n}} & x_{2\phantom{-n}} & x_1
    \end{matrix}
\end{align}
\hspace{5.7cm} $\mathbf{x}_1$ \hspace{2.5mm} $\mathbf{x}_2$ \hspace{.0mm} $\dots$ \hspace{-0.5mm}  $\mathbf{x}_{n-1}$ \hspace{-0.5mm}  $\mathbf{x}_n$ \\
\\
For  $n = 2^t$, a full-rank matrix can be constructed by selecting an initial vector with odd weight. This ensures that the associated polynomial  $f$ has no root at $x = 1$, meaning $ x - 1$ does not divide $f$. Consequently, the greatest common divisor of $f$ and $x^n - 1 = x^{2^t} - 1 = (x - 1)^{2^t}$ is 1. By Theorem \ref{oddWeight}, this guarantees that the resulting matrix is full rank.

\begin{prop}
   The circulant matrix of dimension $n=2^t$ obtained by rotating an initial vector of odd weight is complete (i.e. full rank).  
\end{prop}
Proving the unconditional binding property of the new commitment scheme constructed in this paper requires one last result:
\begin{theo}\label{LinIndThm}
    Let $z$ be a power of 2 greater than $2t$, $\mathbf{b} \ne \mathbf{0}$, $\mathbf{s} = \sum_{i=1}^z \hat{b}_i \cdot \mathbf{r}_i$, where $\hat{b}_i = b_i$ for $i = 1,\dots,t$ and $\hat{b}_i = 0$ otherwise, and $\{\mathbf{r}_i\}_{i=1}^z$ are vectors generated with algorithm  in \Cref{finalAlg} $($namely, $\mathbf{r}_i=\mathbf{x}_i)$, where $\mathbf{r}_1$ has odd weight. Then, the matrix representing the linear system
    \[
    \begin{cases}
        s_1 =  \sum_{k=1}^z b_k \cdot x_{1+k-1 \mod z} \\
        \hspace{0.5cm}\vdots \\
        s_z = \sum_{k=1}^z b_k \cdot x_{z+k-1 \mod z}
    \end{cases}
    \]
    has at least $z/2$ linearly independent vectors.
\end{theo}
\begin{proof}
    The matrix representing the linear system is a circulant matrix. Moreover, in the first row, there is a 1 followed by at least $z/2$ consecutive zeros. If we consider the $z/2$ vectors obtained by rotating this vector one bit to the right each time, they can be grouped and seen as forming a lower triangular matrix with 1s on the diagonal. Therefore, they are linearly independent.
\end{proof}
With these properties, the commitment scheme described in \Cref{ourBSC2fin} can be proved to be statistically binding. 
\begin{figure}
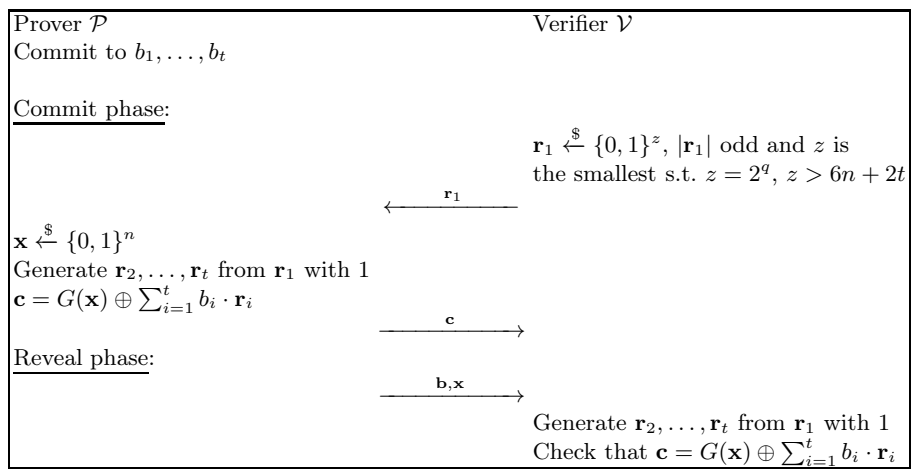

        \centering
        \resizebox{\columnwidth}{!}{
        \begin{tabular}{|l c l|}
            \hline
            Prover $\mathcal{P}$ & & Verifier $\mathcal{V}$ \\
            Commit to $b_1, \dots, b_t$ & & \\
            & & \\
            \underline{Commit phase}: & & \\
            & & $\mathbf{r}_1 \xleftarrow{\$} \{ 0,1 \} ^{z}$, $|\mathbf{r}_1|$ odd and $z$ is\\
            & &the smallest s.t. $z=2^q$, $z>6n+2t$ \\
            & $\xleftarrow{\hspace{20 pt}  \mathbf{r}_1 \hspace{20 pt}}$ & \\
            $\mathbf{x} \xleftarrow{\$} \{ 0,1 \} ^{n}$ & &\\
            Generate $\mathbf{r}_2, \dots, \mathbf{r}_{t}$ from $\mathbf{r}_1$ with \ref{finalAlg} & &\\
            $\mathbf{c} = G(\mathbf{x}) \oplus \sum_{i=1}^{t} b_i \cdot \mathbf{r}_i$ & &\\
            & $\xrightarrow{\hspace{24 pt} \mathbf{c} \hspace{24 pt}}$ & \\
            \underline{Reveal phase}: & & \\
            & $\xrightarrow{\hspace{20 pt} \mathbf{b}, \mathbf{x} \hspace{20 pt}}$ & \\
            & & Generate $\mathbf{r}_2, \dots, \mathbf{r}_{t}$ from $\mathbf{r}_1$ with \ref{finalAlg} \\
            & & Check that $\mathbf{c} = G(\mathbf{x}) \oplus \sum_{i=1}^{t} b_i \cdot \mathbf{r}_i$ \\
        \hline
        \end{tabular}
        }
        \caption{Our final extension of Naor's bit commitment scheme.}
        \label{ourBSC2fin}
\end{figure}

Suppose indeed the prover wants to commit to a $t$-bit string $b_1, \dots, b_t$. The verifier then sends the initial vector $\mathbf{r}_1$, with odd weight, which is the smallest power of 2 longer than $6n + 2t$. Let $z$ be this length. The prover will encode the string it wants to commit to into a $z$-bit string by appending zeros to the original string. The procedure then continues in the same way as in the previous protocol. If the prover wants to open the string $\mathbf{b}$ to a different one $\mathbf{b}'$, it must find an $\mathbf{x}'$ such that $\mathbf{s} = \sum_{i=1}^z (\hat{b}_i \oplus \hat{b}_i') \mathbf{r}_i = G(\mathbf{x}) \oplus G(\mathbf{x}')$. Thus, $\mathbf{s}$ should belong to the set $\{ G(\mathbf{x}) \oplus G(\mathbf{x}') \mid \mathbf{x}, \mathbf{x}' \in \{0,1\}^n \}$, which has at most $2^{2n}$ elements. Now, thanks to \Cref{LinIndThm} we know that at least half of the bits of $\mathbf{s}$ are uniformly distributed. 
Specifically, for every linear combination (this time given by the coefficients $\hat{b}_i \oplus \hat{b}_i'$, i.e. for every possible value the prover tries to open the commitment to) $\mathbf{s}$ lays in a vector subspace of dimension at least $z/2 > (6n+2t)/2$. This means that such an $\mathbf{x}'$ exists with a probability of at most \[\frac{2^{2n}}{2^{(6n+2t)/2}} = \frac{2^{2n}}{2^{3n+t}}.\] Thus, the final probability of cheating can be bounded by $2^{-n}$, as before.
The computational hiding property comes directly from Naor's commitment scheme.
The communication complexity of this scheme is $2z+n+t$, where $z$ is the smallest power of 2 strictly greater than $6n+2t$. This complexity is significantly easier to quantify compared to that of Naor's string commitment scheme.

\subsection{Extending Kilian with Our Approach}
Starting from the natural extension of Naor and from the Kilian method we can devise a new scheme whose flow is described in \Cref{ourBSC+Kilian}. In particular, here we have two pseudorandom number generators $G_1:\{0,1\}^n \rightarrow \{0,1\}^{z}$ and $G_2: \{0,1\}^l \rightarrow \{ 0,1 \}^t$.
\begin{figure}[!h]
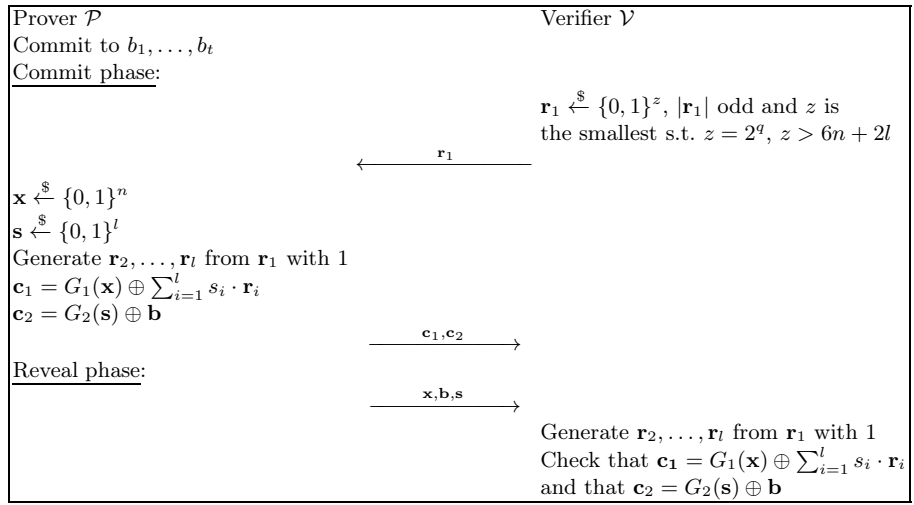

        \centering
        \resizebox{\columnwidth}{!}{
        \begin{tabular}{|l c l|}
            \hline
            Prover $\mathcal{P}$ & & Verifier $\mathcal{V}$ \\
            Commit to $b_1, \dots, b_t$ & & \\
            \underline{Commit phase}: & & \\
            & & $\mathbf{r}_1 \xleftarrow{\$} \{ 0,1 \} ^{z}$, $|\mathbf{r}_1|$ odd and $z$ is\\
            & &the smallest s.t. $z=2^q$, $z>6n+2l$ \\
            & $\xleftarrow{\hspace{30 pt}  \mathbf{r}_1 \hspace{30 pt}}$ & \\
            $\mathbf{x} \xleftarrow{\$} \{ 0,1 \} ^{n}$ & &\\
            $\mathbf{s} \xleftarrow{\$}\{0,1\}^l$ & & \\
            Generate $\mathbf{r}_2, \dots, \mathbf{r}_{l}$ from $\mathbf{r}_1$ with \ref{finalAlg} & &\\
            $\mathbf{c}_1 = G_1(\mathbf{x}) \oplus \sum_{i=1}^l s_i \cdot \mathbf{r}_i$ & &\\
            $\mathbf{c}_2 = G_2(\mathbf{s}) \oplus \mathbf{b}$ & & \\
            & $\xrightarrow{\hspace{20 pt} \mathbf{c}_1, \mathbf{c}_2 \hspace{20 pt}}$ & \\
            \underline{Reveal phase}: & & \\
            & $\xrightarrow{\hspace{20 pt} \mathbf{x}, \mathbf{b}, \mathbf{s} \hspace{20 pt}}$ & \\
            & & Generate $\mathbf{r}_2, \dots, \mathbf{r}_{l}$ from $\mathbf{r}_1$ with \ref{finalAlg}\\
            & & Check that $\mathbf{c_1} = G_1(\mathbf{x}) \oplus \sum_{i=1}^l s_i \cdot \mathbf{r}_i$ \\
            & & and that $\mathbf{c}_2 = G_2(\mathbf{s}) \oplus \mathbf{b}$ \\
        \hline
        \end{tabular}
        }
        \caption{Natural extension of Naor's bit CS + Kilian.}
        \label{ourBSC+Kilian}
    \end{figure}
In this approach, the Kilian technique is employed by committing to a seed of a pseudorandom number generator instead of the entire bit string. The bit string is then masked by XORing it with another bit string generated by a pseudorandom number generator, using the previously committed seed. The key difference from the scheme described by Kilian in \cite{JC:Naor91} lies in the method of seed commitment. Specifically, the seed is committed with the natural extension of Naor described in \Cref{ourBSC2fin}. Recall that this scheme is more efficient compared to the normal extension when the string to commit to is much greater than the length of a secure seed of a PRNG.

The proofs of computational hiding and statistical binding properties naturally follow the proofs of our extension of Naor and the Kilian method.

\subsection{The Commitment Scheme Used in qOT}
The commitment scheme used in the qOT protocol is shown in \Cref{2BSC}. Observe that this does not follow the rules of the general string commitment scheme described in \Cref{ourBSC2fin}. It is an ad-hoc version for committing to 2-bit strings.
\begin{figure}[!h]
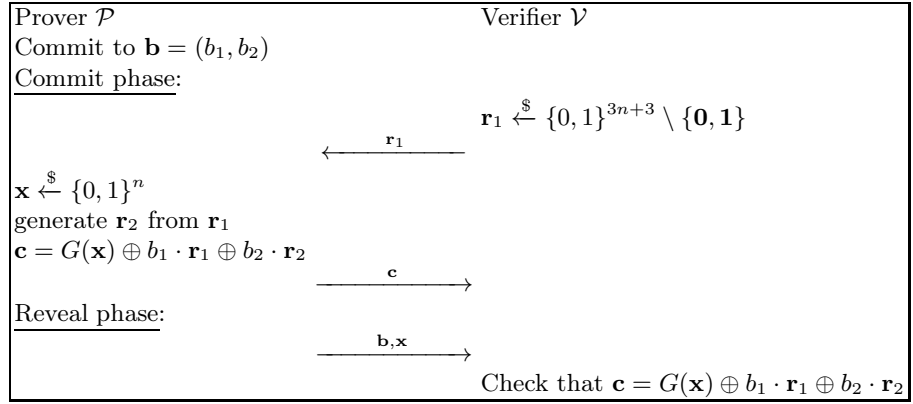

        \centering
        \resizebox{\columnwidth}{!}{
        \begin{tabular}{|l c l|}
            \hline
            Prover $\mathcal{P}$ & & Verifier $\mathcal{V}$ \\
            Commit to $\mathbf{b}=(b_1, b_2)$ & & \\
            \underline{Commit phase}: & & \\
            & & $\mathbf{r}_1 \xleftarrow{\$} \{ 0,1 \} ^{3n+3} \setminus \{\mathbf{0}, \mathbf{1}\}$ \\
            & $\xleftarrow{\hspace{20 pt}  \mathbf{r}_1 \hspace{20 pt}}$ & \\
            $\mathbf{x} \xleftarrow{\$} \{ 0,1 \} ^{n}$ & &\\
            generate $\mathbf{r}_2$ from $\mathbf{r}_1$ & &\\
            $\mathbf{c} = G(\mathbf{x}) \oplus b_1 \cdot \mathbf{r}_1 \oplus b_2 \cdot \mathbf{r}_2$ & &\\
            & $\xrightarrow{\hspace{24 pt} \mathbf{c} \hspace{24 pt}}$ & \\
            \underline{Reveal phase}: & & \\
            & $\xrightarrow{\hspace{20 pt} \mathbf{b}, \mathbf{x} \hspace{20 pt}}$ & \\
            & & Check that $\mathbf{c} = G(\mathbf{x}) \oplus b_1 \cdot \mathbf{r}_1 \oplus b_2 \cdot \mathbf{r}_2$ \\
        \hline
        \end{tabular}
        }
        \caption{Extension of Naor's bit CS for 2-bit strings.}
        \label{2BSC}
    \end{figure}
Initially, the verifier generates a random vector $\mathbf{r}_1$ of length $3n + 3$, ensuring that it differs from both the all-zero vector and the all-one vector. This vector $\mathbf{r}_1$ is then sent to the prover. In response, the prover selects a random seed $\mathbf{x}$ of length $n$ for the pseudorandom number generator and generates $\mathbf{r}_2$ by rotating $\mathbf{r}_1$ by one position. The commitment is subsequently computed following the same procedure as in the commitment scheme described in \Cref{ourBSC2fin}, and the verification phase proceeds in the same manner.

Clearly, the hiding property of this scheme follows from the security of the pseudorandom number generator. The binding property instead, can be proved to hold statistically. Suppose indeed the prover has already committed to $\mathbf{b}$ and wants to open the commitment to a different string. Then it should find an $\mathbf{x}'$ such that \[(b_1 \oplus b_1')\mathbf{r}_1 \oplus (b_2 \oplus b_2')\mathbf{r}_2 = G(\mathbf{x}) \oplus G(\mathbf{x'})\] for a $\mathbf{b}' \ne \mathbf{b}$. As before, the set $\{G(\mathbf{x}) \oplus G(\mathbf{x}') |\mathbf{x},\mathbf{x}' \in \{0,1\}^n\}$ has at most $2^{2n}$ elements. Now, there are three cases:
\begin{itemize}
    \item if $\mathbf{b}' \oplus \mathbf{b} = (1,0)$ then the above equation is \[ \mathbf{r}_1 = G(\mathbf{x}) \oplus G(\mathbf{x'})\]
    and since $\mathbf{r}_1$ is randomly chosen, the probability that such an $\mathbf{x'}$ exists is \[ \frac{2^{2n}}{2^{3n+3}-2};\]
    \item if $\mathbf{b}' \oplus \mathbf{b} = (1,0)$ then we have the same as before with $\mathbf{r}_2$ in place of $\mathbf{r}_1$, but $\mathbf{r}_2$ is obtained by rotating $\mathbf{r}_1$, so it is random as well, and the same as before can be said;
    \item if instead $\mathbf{b}' \oplus \mathbf{b} = (1,1)$ then the equation the prover must be able to satisfy is \[ \mathbf{r}_1 \oplus \mathbf{r}_2= G(\mathbf{x}) \oplus G(\mathbf{x'}).\] Now, for every value of $\mathbf{r}_1$, $\mathbf{r}_1 \oplus \mathbf{r}_2$ has even weight. However, there are no more constraints on it. So, it is a random value in a subspace of dimension $3n+3-1$, and the probability of finding an $\mathbf{x}'$ satisfying the above equation is \[\frac{2^{2n}}{2^{3n+3-1}-1}\] (it is not possible that $\mathbf{r}_1 \oplus \mathbf{r}_2 = \mathbf{0}$ because it would imply $\mathbf{r}_1 = \mathbf{r}_2$ and, in our case, this implies either $\mathbf{r}_1 = \mathbf{0}$ or $\mathbf{r}_1=\mathbf{1}$.) 
\end{itemize}
In the end, there is a probability of at most $\frac{2^{2n}}{2^{3n+3-1} -1} \simeq 2^{-n-2}$ that an $\mathbf{x}'$ allowing the prover to open the commitment to another specific string exists (the above estimation can be considered as an equality since the value of $n$, as length of the seed of the PRNG, should be at least 128, but also 256 in quantum environments, to bring a good enough level of security regarding the hiding property. So, the effect of subtracting 1 from the denominator is negligible). In the end, the probability that there exists an $\mathbf{x}'$ that allows the prover to open the commitment to any other string is bounded by $4 \cdot 2^{-n-2} = 2^{-n}$.  With this approach, the communication complexity is $7n+8$ bits for each 2-bit string to be committed to.

\section{A New Pre-Processing Model for Commitments}
Despite our improvements, for qOT a commitment scheme satisfying the subvector opening property is needed.
This means that the prover should be able to open only some of the committed bits.
However, we are not aware of any efficient vector commitment scheme (where the prover can indeed open only a subset of committed bits) or, in general, string commitment scheme with the subvector opening property satisfying the unconditional binding property based on the existence of OWF only.
In this situation, one has to commit to all the bits (or all the measurements, i.e. couple of bits) separately.
This is very time-consuming, considering that qOT requires a lot of quantum communication (so a lot of states to be measured and a lot of measurements to be committed).

We therefore introduce a new preprocessing model to improve the online phase of commitment schemes.
By letting the sender and receiver of the qOT exchange data during a preprocessing phase, i.e., before the actual qOT begins, the receiver could commit to random bits.
During the actual qOT execution, to commit to its measurements, the prover can then simply XOR the bits it wants to commit to with the bits already committed in the preprocessing phase.
To open a specific commitment, the prover reveals the commitment from the preprocessing phase and discloses the committed bit.
In this way, the verifier can check that the bit committed to in the preprocessing phase meets the commitment's conditions and can verify that the XOR result is as expected.

This method is extremely fast and efficient during the qOT execution. Additionally, its security is equivalent to the security of the commitment scheme used in the preprocessing phase.
For example, suppose the prover committed to a bit $m$ in the preprocessing phase using an unconditionally binding commitment scheme, such as Naor's scheme. To commit to a new bit $b$, the prover simply publishes $b \oplus m$. Then, to open the commitment, the prover reveals $b$ and everything needed to reveal the commitment generated with $m$. The flow of the protocol is described in \Cref{NaorBCwithPre}.
    \begin{figure}[!h]
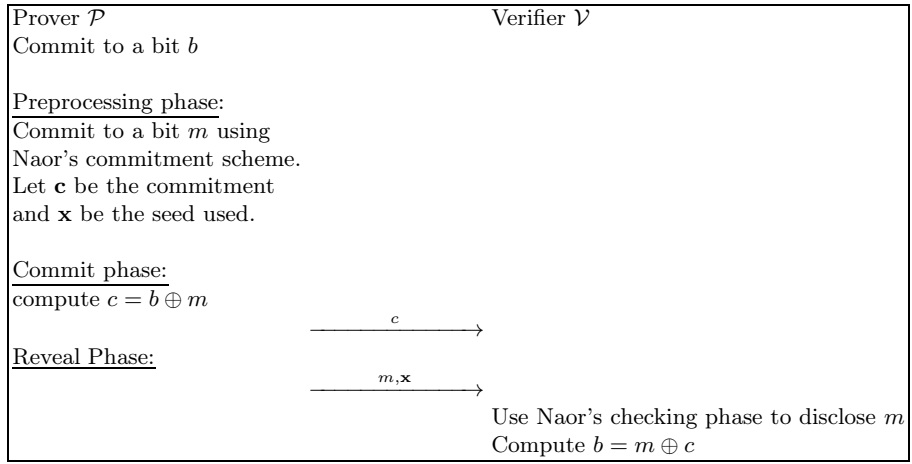

        \centering
        \resizebox{\columnwidth}{!}{
        \begin{tabular}{|l c l|}
            \hline
            Prover $\mathcal{P}$ & & Verifier $\mathcal{V}$ \\
            Commit to a bit $b$ & & \\
            & & \\
             \underline{Preprocessing phase}: & & \\
             Commit to a bit $m$ using & & \\
             Naor's commitment scheme. & & \\
             Let $\mathbf{c}$ be the commitment & & \\
             and $\mathbf{x}$ be the seed used. & & \\
             & & \\
            \underline{Commit phase:} & & \\
            compute $c= b \oplus m$ & & \\
            & $\xrightarrow{\hspace{30pt} c \hspace{30pt}}$ & \\
            \underline{Reveal Phase:} & & \\
            & $\xrightarrow{\hspace{25pt} m, \mathbf{x} \hspace{25pt}}$ & \\
            & & Use Naor's checking phase to disclose $m$ \\
            & & Compute $b = m \oplus c$ \\
        \hline
        \end{tabular}
        }
        \caption{Naor's bit string commitment protocol with preprocessing phase.}
        \label{NaorBCwithPre}
    \end{figure}

The properties of the commitment execution following this ``trick'' are the same as the commitment scheme used in the preprocessing phase:
\begin{itemize}
    \item \textit{Binding}: suppose the prover commits to $b$ by publishing $b \oplus m$ for a committed $m$ in the preprocessing phase, with commitment $c$. To open to $1 \oplus b$, the prover must be able to open $c$ to $1 \oplus m$, so that $(1 \oplus m) \oplus (1 \oplus b) = m \oplus b$. However, if the prover cannot open $c$ to $1 \oplus m$, this is impossible. Therefore, to break the binding property here, the prover would need to break the binding property of the commitment scheme used in the preprocessing phase.
    \item \textit{Hiding}: again, suppose the prover commits to $b$ by publishing $b \oplus m$ for a committed $m$ in the preprocessing phase, with commitment $c$. For the verifier, learning $b$ from $b \oplus m$ is equivalent to learning $m$ from $b \oplus m$. To learn $m$, the verifier would need to break the hiding property of the commitment scheme used in the preprocessing phase.
\end{itemize}
Note that this is a generic proof: it works independently of the commitment scheme considered.

\section{Conclusion}

In conclusion, the optimal choice for the commitment scheme in the quantum oblivious transfer protocol is to commit to each quantum state separately (i.e., 2 bits at a time) with our bit string commitment scheme (\Cref{2BSC}).
This approach offers a significant improvement, considering that most known implementations rely on hash-based string commitments, which are proven secure only under the random oracle model.
By committing to each state individually using our string commitment scheme, we enhance security, as we can prove its security without relying on the random oracle model assumption. Furthermore, with the preprocessing phase technique, we reduce the computational cost of the commitment in the online phase to almost zero, greatly improving efficiency.

Future work may include trying to find an unconditional binding commitment scheme based on OWF with the subvector opening property, as well as trying to improve the efficiency of our new string commitment scheme.

\begin{credits}
\subsubsection{\ackname}
This work has received funding from the European Union’s Horizon Europe research and innovation program under No. 101114043 ("QSNP"), from the DIGITAL-2021-QCI-01 Digital European Program under Project number No 101091642 and the National Foundation for Research, Technology and Development ("QCI-CAT"). This work was supported in part by the Science for Peace and Security (SPS) NATO Programme, through the project
"QSCAN" (reference SPS.MYP.G6158).
\end{credits}

\bibliography{abbrev3,crypto,bib}

\begin{thebibliography}{10}
\providecommand{\url}[1]{#1}
\csname url@samestyle\endcsname
\providecommand{\newblock}{\relax}
\providecommand{\bibinfo}[2]{#2}
\providecommand{\BIBentrySTDinterwordspacing}{\spaceskip=0pt\relax}
\providecommand{\BIBentryALTinterwordstretchfactor}{4}
\providecommand{\BIBentryALTinterwordspacing}{\spaceskip=\fontdimen2\font plus
\BIBentryALTinterwordstretchfactor\fontdimen3\font minus \fontdimen4\font\relax}
\providecommand{\BIBforeignlanguage}[2]{{%
\expandafter\ifx\csname l@#1\endcsname\relax
\typeout{** WARNING: IEEEtranS.bst: No hyphenation pattern has been}%
\typeout{** loaded for the language `#1'. Using the pattern for}%
\typeout{** the default language instead.}%
\else
\language=\csname l@#1\endcsname
\fi
#2}}
\providecommand{\BIBdecl}{\relax}
\BIBdecl

\bibitem{scott_quantum_2002}
\BIBentryALTinterwordspacing
S.~Aaronson, ``Quantum lower bound for the collision problem,'' in \emph{Proceedings of the Thiry-Fourth Annual ACM Symposium on Theory of Computing}, ser. STOC '02, New York, NY, USA, 2002, p. 635–642. [Online]. Available: \url{https://doi.org/10.1145/509907.509999}
\BIBentrySTDinterwordspacing

\bibitem{AC:ABBCP13}
M.~Abdalla, F.~Benhamouda, O.~Blazy, C.~Chevalier, and D.~Pointcheval, ``{SPHF}-friendly non-interactive commitments,'' in \emph{ASIACRYPT~2013, Part~I}, ser. {LNCS}, K.~Sako and P.~Sarkar, Eds., vol. 8269.\hskip 1em plus 0.5em minus 0.4em\relax Springer, Berlin, Heidelberg, Dec. 2013, pp. 214--234.

\bibitem{10.1007/3-540-46766-1_29}
C.~H. Bennett, G.~Brassard, C.~Cr{\'e}peau, and M.-H. Skubiszewska, ``Practical quantum oblivious transfer,'' in \emph{Advances in Cryptology --- CRYPTO '91}, J.~Feigenbaum, Ed.\hskip 1em plus 0.5em minus 0.4em\relax Berlin, Heidelberg: Springer Berlin Heidelberg, 1992, pp. 351--366.

\bibitem{faz_hernandes_performance_2018}
\BIBentryALTinterwordspacing
A.~Faz-Hernandez, J.~L\'{o}pez, and A.~K. D.~S. de~Oliveira, ``Sok: A performance evaluation of cryptographic instruction sets on modern architectures,'' in \emph{Proceedings of APKC '18}, New York, NY, USA, 2018, p. 9–18. [Online]. Available: \url{https://doi.org/10.1145/3197507.3197511}
\BIBentrySTDinterwordspacing

\bibitem{FOCS:GKMRV00}
Y.~Gertner, S.~Kannan, T.~Malkin, O.~Reingold, and M.~Viswanathan, ``The relationship between public key encryption and oblivious transfer,'' in \emph{41st FOCS}.\hskip 1em plus 0.5em minus 0.4em\relax {IEEE} Computer Society Press, Nov. 2000, pp. 325--335.

\bibitem{halevi_practical_1996}
S.~Halevi and S.~Micali, ``\BIBforeignlanguage{en}{Practical and {Provably}-{Secure} {Commitment} {Schemes} from {Collision}-{Free} {Hashing}},'' in \emph{\BIBforeignlanguage{en}{Advances in {Cryptology} — {CRYPTO} ’96}}, N.~Koblitz, Ed.\hskip 1em plus 0.5em minus 0.4em\relax Springer, 1996, pp. 201--215.

\bibitem{AC:JKPT12}
A.~Jain, S.~Krenn, K.~Pietrzak, and A.~Tentes, ``Commitments and efficient zero-knowledge proofs from learning parity with noise,'' in \emph{ASIACRYPT~2012}, ser. {LNCS}, X.~Wang and K.~Sako, Eds., vol. 7658.\hskip 1em plus 0.5em minus 0.4em\relax Springer, Berlin, Heidelberg, Dec. 2012, pp. 663--680.

\bibitem{CCS:KKRT16}
V.~Kolesnikov, R.~Kumaresan, M.~Rosulek, and N.~Trieu, ``Efficient batched oblivious {PRF} with applications to private set intersection,'' in \emph{ACM CCS 2016}, E.~R. Weippl, S.~Katzenbeisser, C.~Kruegel, A.~C. Myers, and S.~Halevi, Eds.\hskip 1em plus 0.5em minus 0.4em\relax {ACM} Press, Oct. 2016, pp. 818--829.

\bibitem{lemus2025performancepracticalquantumoblivious}
\BIBentryALTinterwordspacing
M.~Lemus, P.~Schiansky, M.~Goulão, M.~Bozzio, D.~Elkouss, N.~Paunković, P.~Mateus, and P.~Walther, ``Performance of practical quantum oblivious key distribution,'' 2025. [Online]. Available: \url{https://arxiv.org/abs/2501.03973}
\BIBentrySTDinterwordspacing

\bibitem{DBLP:journals/cacm/Lindell21}
Y.~Lindell, ``Secure multiparty computation,'' \emph{Commun. {ACM}}, vol.~64, no.~1, pp. 86--96, 2021.

\bibitem{10.1007/1-4020-8143-X_29}
L.~L{\'i}{\v{s}}kov{\'a} and M.~Stanek, ``Efficient simultaneous contract signing,'' in \emph{Security and Protection in Information Processing Systems}, Y.~Deswarte, F.~Cuppens, S.~Jajodia, and L.~Wang, Eds.\hskip 1em plus 0.5em minus 0.4em\relax Boston, MA: Springer US, 2004, pp. 441--455.

\bibitem{JC:Naor91}
M.~Naor, ``Bit commitment using pseudorandomness,'' \emph{Journal of Cryptology}, vol.~4, no.~2, pp. 151--158, Jan. 1991.

\bibitem{STOC:NaoPin99}
M.~Naor and B.~Pinkas, ``Oblivious transfer and polynomial evaluation,'' in \emph{31st ACM STOC}.\hskip 1em plus 0.5em minus 0.4em\relax {ACM} Press, May 1999, pp. 245--254.

\bibitem{santos_private_2022}
\BIBentryALTinterwordspacing
M.~B. Santos, A.~C. Gomes, A.~N. Pinto, and P.~Mateus, ``Private {Computation} of {Phylogenetic} {Trees} {Based} on {Quantum} {Technologies},'' \emph{IEEE Access}, vol.~10, pp. 38\,065--38\,088, 2022. [Online]. Available: \url{https://ieeexplore.ieee.org/document/9732453}
\BIBentrySTDinterwordspacing

\bibitem{santos_quantum_2022}
\BIBentryALTinterwordspacing
M.~B. Santos, P.~Mateus, and A.~N. Pinto, ``\BIBforeignlanguage{en}{Quantum oblivious transfer: a short review},'' \emph{\BIBforeignlanguage{en}{Entropy}}, vol.~24, no.~7, p. 945, Jul. 2022, arXiv:2206.03313 [quant-ph]. [Online]. Available: \url{http://arxiv.org/abs/2206.03313}
\BIBentrySTDinterwordspacing

\end{thebibliography}
\bibliographystyle{IEEEtranS}
\end{document}